\pdfoutput=1
\documentclass[12pt, oneside]{amsart}

\usepackage{hyperref}
\hypersetup{%
	pdfpagemode={UseOutlines},
	bookmarksopen,
	pdfstartview={FitH},
	colorlinks,
	linkcolor={blue},
	citecolor={blue},
	urlcolor={blue}
}

\usepackage{geometry} 

\usepackage{dcolumn}
\newcolumntype{d}{D{.}{.}{-1} } 

\usepackage{amsmath, amssymb, amsthm}
\usepackage{mathrsfs}
\usepackage{enumitem}
\usepackage{aligned-overset}
\usepackage{mathtools}
\usepackage[numbers]{natbib}
\usepackage{bclogo}

\numberwithin{equation}{section}

\newtheorem{teo}{Theorem}[section]

\newtheorem{lem}[teo]{Lemma}

\theoremstyle{remark}
\newtheorem{remark}[teo]{Remark}

\newcommand{\R}{\mathbb{R}} 
\newcommand{\N}{\mathbb{N}} 
\renewcommand{\P}{\mathcal{P}} 
\newcommand{\C}{\mathbb{C}} 
\newcommand{\commF}[1]{{\color{red}*** #1 ***}}
\newcommand{\commFe}[1]{{\color{blue}*** #1 ***}}
\newcommand{\commP}[1]{{\color{orange}*** #1 ***}}

\title[An elementary approach to Wehrl-type entropy bound]{An elementary approach to Wehrl-type entropy bounds in quantitative form}
\author{Fabio Nicola}
\address[Fabio Nicola]{Dipartimento di Scienze Matematiche, Politecnico di Torino, Corso Duca degli Abruzzi 24, 10129 Torino, Italy}
\email{fabio.nicola@polito.it}

\author{Federico Riccardi}
\address[Federico Riccardi]{Dipartimento di Scienze Matematiche, Politecnico di Torino, Corso Duca degli Abruzzi 24, 10129 Torino, Italy}
\email{federico.riccardi@polito.it}

\author{Paolo Tilli}
\address[Paolo Tilli]{Dipartimento di Scienze Matematiche, Politecnico di Torino, Corso Duca degli Abruzzi 24, 10129 Torino, Italy}
\email{paolo.tilli@polito.it}
\begin{document}

	\keywords{Wehrl entropy, coherent states, representations, holomorphic polynomials, stability, Bergman spaces}
	\subjclass[2020]{81R30, 22E70, 30H20, 32A10, 47N50, 49J40}
	
	\begin{abstract}
We consider the problem of the stability (with sharp exponent) of the Lieb--Solovej inequality for symmetric $SU(N)$ coherent states, which was obtained only recently by the authors. Here, we propose an elementary proof of this result, based on reformulating the Wehrl-type entropy as a function defined on the unit sphere in $\mathbb{C}^d$, for some suitable $d$, and on some explicit (and somewhat surprising) computations.
\end{abstract}

	\maketitle
	
	\section{Introduction}
    Recent years have witnessed an intense and fruitful interest in various questions connected with the \emph{Wehrl-type entropy conjecture}. The origins of this problem  date back to the late 1970s, when Wehrl \cite{wehrl_entropy} introduced a notion of classical entropy for quantum-mechanical density matrices on $L^2(\R)$ and conjectured that Glauber coherent states, which are coherent states in representations of the Heisenberg group, are the only optimizers for this entropy. Within a short period, Lieb \cite{lieb_entropy} proved that coherent states are indeed optimizers (with uniqueness
     established later by Carlen \cite{carlen}) and suggested that an analogous result should hold also for Bloch coherent spin states, which are coherent states for the irreducible representations of $SU(2)$. This new conjecture was proven in full generality several years later by Lieb and Solovej \cite{lieb_solovej_SU(2)}, who further generalized their proof \cite{lieb_solovej_SU(N)} in order to show an analogous result for the symmetric $SU(N)$ coherent states ($N \geq 2$) (see also \cite{nicolatilli_fk, frank_sharp_inequalities} for the full conjecture on Glauber coherent states and \cite{kulikov} for $SU(1,1)$ coherent states). However, the question of uniqueness of the optimizers for the $SU(2)$ and the $SU(N)$ Wehrl-type entropy remained open until recently. For the $SU(2)$ case, it was proved independently and simultaneously in \cite{frank_sharp_inequalities} and in \cite{kulikov_nicola_ortega_tilli}, while the $SU(N)$ case was finally settled in \cite{NRT_wehrl_SU(N)}. Along with the question of uniqueness, the issue of stability was also addressed in recent years. Roughly speaking, the goal of a stability estimate is to prove and quantify the idea that states with almost minimal entropy should be almost coherent states. The stability of the Lieb--Solovej inequality for $SU(2)$ coherent states was proved in \cite{garcia_ortega_stability}, while for symmetric $SU(N)$ coherent states it was proved in \cite{NRT_wehrl_SU(N)} (see also \cite{frank2023generalized, GGRT, gomez2024uniform, melentijevic2025stability} for the stability of the estimate in the case of Glauber and $SU(1,1)$ coherent states). In all these cases, the proof relies on a careful and highly nontrivial study of the measure of the super-level sets of the Husimi function (see below) using a technique first introduced in \cite[][Lemma 2.1]{GGRT}.
    
    The main goal of this paper is to prove the stability of the estimate for the Wehrl-type entropy conjecture for symmetric $SU(N)$ coherent states using only elementary tools and avoiding any study of the measure of the super-level sets of the Husimi function. In fact, to simplify matters and avoid additional technicalities, we will focus on the case of pure states, although this method could likely be adapted to handle density matrices as well. The main idea of the proof is to rephrase the problem in order to express the Wehrl-type entropy as a function on the unit sphere in a high-dimensional Euclidean space. Then, we are able to compute the second differential of the entropy on the submanifold where the entropy achieves its minimum value, and we can show that it is uniformly positive definite on the normal bundle to the submanifold (in fact, we will instead consider an equivalent formulation with a convex function rather than a concave one, which corresponds to a maximization problem).
    
    Although this approach is elementary, it requires some lengthy---and somewhat surprising---computations. In order to make things clearer, we have organized the paper as follows.

    In Section \ref{sec:notation and setting} we fix the notation and rephrase the problem in order to write the Wehrl-type entropy as a function on the unit sphere in $\C^d$, for a suitable 
    dimension $d$. This geometric setting requires some preliminary notions of differential geometry, which are gathered at the end of the section. In Section \ref{sec:proof regular case} we prove the stability result under an additional smoothness assumption (which is satisfied in the special cases of interest; see Remark \ref{remark: xlogx} below). This enables us to present the core concept of the proof without encountering further technical complications. The proof of the result in full generality is then given in Section \ref{sec:proof non regular case} and will require some technical (and more cumbersome) lemmas that are collected in Appendix \ref{sec:appendix}.
	
	\section{Notation and rephrasing of the problem in terms of holomorphic polynomials}\label{sec:notation and setting}

    \subsection{Notation}
    For $z = (z_1,\ldots,z_N),\, w = (w_1,\ldots,w_N)  \in \C^N$ and a multi-index $\alpha = (\alpha_1, \ldots, \alpha_N) \in \N^N$ we adopt the standard notation
	\begin{gather*}
		z^\alpha = z_1^{\alpha_1} \cdots z_N^{\alpha_N}, \quad |z|^2 = |z_1|^2 + \cdots +|z_N|^2, \quad z \cdot w = z_1w_1 + \cdots +z_N w_N,\\
		|\alpha| = \alpha_1 + \cdots +\alpha_N, \quad \alpha! = \alpha_1! \cdots \alpha_N!
	\end{gather*}
    Given $M \in \N$ and a multi-index $\alpha \in \N^N$ of length at most $M$ we have
    \begin{equation*}
        \binom{M}{\alpha} = \dfrac{M!}{\alpha_1! \cdots \alpha_N! (M-|\alpha|)!}.
    \end{equation*}
    Dealing with sums over a multi-index, we denote by
    $$ \sum_{|\alpha| \leq M}$$
     the sum over all  multi-indices $\alpha \in \N^N$ of total length at most $M$, and by
    \begin{equation*}
        \sum_{\alpha \neq 0}
    \end{equation*}
the sum over all  multi-indices of total length at most $M$ except for the null multi-index,
the condition $|\alpha|\leq M$ being understood
(this should cause no confusion, since $M$ is fixed throughout the paper).
Similarly, we denote by
    \begin{equation*}
        \sum_{\alpha \neq \beta}\
    \end{equation*}
    the double sums over all multi-indices $\alpha$ and $\beta$ (of length at most $M$) that are not equal (i.e. they differ in at least one component).
    \subsection{From the original setting to the unit sphere}
	The original theorem by Lieb and Solovej \cite{lieb_solovej_SU(N)} is stated in terms of symmetric representations of $SU(N+1)$ as follows. Fix $N \geq 1$, $M \geq 1$ and consider the Hilbert space $\mathcal{H}_M = (\bigotimes^M \C^{N+1})_{\text{sym}}$, that is the image of $\bigotimes^M \C^{N+1}$ under the projection
	\begin{equation*}
		P_M (v_1 \otimes \cdots \otimes v_M) = \dfrac{1}{M!} \sum_{\sigma \in S_M} v_{\sigma(1)} \otimes \cdots \otimes v_{\sigma(M)},
	\end{equation*}
	where $S_M$ denotes the permutation group over $\{1,\ldots, M\}$. We denote by $\langle \psi|\phi\rangle$ the inner product of $\psi,\phi\in\mathcal{H}_M$, with the agreement that it is linear in the second argument. On $\mathcal{H}_M$, we consider the representation of the group $SU(N+1)$ given by
	\begin{equation*}
		\begin{aligned}
			R(v_1 \otimes \cdots \otimes v_M) = (Rv_1) \otimes \cdots \otimes (Rv_M),\qquad R\in SU(N+1)
		\end{aligned}
	\end{equation*}
	and extended on the whole $\mathcal{H}_M$ by linearity. Fix a normalized vector $v \in \C^{N+1}$. Given a density matrix $\rho$ (that is, a nonnegative self-adjoint operator on $\mathcal{H}_M$ with trace one), its \emph{Husimi function} $u_\rho:SU(N+1)\to\mathbb{R}$ is defined as
	\begin{equation*}
		u_{\rho}(R) = \langle \otimes^M Rv \vert \rho \vert \otimes^M R v \rangle, \quad R \in SU(N+1).
\end{equation*}
It is easy to see (\cite{lieb_solovej_SU(N)}) that $u_\rho$ is continuous (in fact real-analytic), and satisfies $0\leq u_\rho\leq 1$, as well as 
\begin{equation}\label{eq:normalization of the Husimi function}
    \textrm{dim}(\mathcal{H}_M)\int_{SU(N+1)} u_\rho(R)\, dR=1,
\end{equation}
where $dR$ denotes the Haar probability measure on $SU(N+1)$.

	Then, the Lieb--Solovej inequality \cite[][Theorem 2.1]{lieb_solovej_SU(N)} states that, for any convex function $\Phi \colon [0,1] \to \R$, the Wehrl-type entropy
	\begin{equation*}
		G(\rho) = \int_{SU(N+1)} \Phi(u_{\rho}(R)) \, dR
	\end{equation*}
	is maximized by (symmetric) coherent states, that is, when \[
    \rho = \vert \otimes^M v \rangle \langle \otimes^M v|
    \]
    for some $v \in \C^{N+1}$ of norm one. In fact, with a slight abuse of language, we will call a coherent state any vector of the type $\otimes^M v\in \mathcal{H}_M$, with $v\in\mathbb{C}^{N+1}$, $|v|=1$.
    \begin{remark}
        Note that any coherent state can be written as $\otimes^M Rv_0$ for some $R \in SU(N+1)$, where $v_0 = (1,0,\ldots,0) \in \C^{N+1}$. In other words, the subset of coherent states is the orbit in $\mathcal{H}_M$ of the coherent state associated with the vector $v_0$ under the action of $SU(N+1)$. 
    \end{remark}
    Note that, when $\Phi$ is affine, $G$ is constant (as easily seen from the normalization \eqref{eq:normalization of the Husimi function}), whereas if $\Phi$ is strictly convex in some interval $(a,1)$ with $a \in (0,1)$ then coherent states are the only maximizers (see \cite[][Theorem 1.2 and Remark 4.3]{NRT_wehrl_SU(N)}). Moreover, the estimate is stable, meaning that  a density matrix $\rho$  that almost maximizes the functional $G$ is almost a coherent state. To make this precise, we introduce the distance, in the trace norm, between $\rho$ and the subset of coherent states, that is,
    \begin{equation*}
        D[\rho] = \inf_{v \in \C^{N+1}, |v|=1} \|\rho - |\otimes^M v \rangle \langle \otimes ^M v|\|_1.
    \end{equation*}
    
    \begin{teo}[\cite{NRT_wehrl_SU(N)}]\label{th:stability in H_M}
        For every convex function $\Phi \colon [0,1] \to \R$ that is strictly convex in some interval $(a,1)$ with $a \in (0,1)$, there exists a constant $c>0$ such that for every density matrix $\rho$ on $\mathcal{H}_{M}$ we have
        \begin{equation*}
            \int_{SU(N+1)} \Phi(u_0(R)) \, dR -  \int_{SU(N+1)} \Phi(u_{\rho}(R)) \, dR \geq c D[\rho]^2,
        \end{equation*}
        where $u_0$ is the Husimi function of any coherent state.
    \end{teo}
    \begin{remark}
        Actually, this theorem holds as long as $\Phi$ is not affine in $(a,1)$ for every $a \in (0,1)$ (see \cite[][Remark 4.4]{NRT_wehrl_SU(N)}). However, for the sake of clarity, we will work under the slightly stronger assumption that $\Phi$ is strictly convex in some interval $(a,1)$, for some $a \in (0,1)$.
    \end{remark}
    
    As already mentioned in the introduction, the goal of this paper is to give an elementary proof of this stability estimate for pure states (i.e. when $\rho = |\psi \rangle \langle \psi|$ for some $\psi \in \mathcal{H}_{M}$). In order to pursue this goal, we need to rephrase the problem in a more geometrical setting. Following the argument presented in \cite[][Section 3.1]{NRT_wehrl_SU(N)} one can show that, for pure states, the Wehrl-type entropy can be written as
	\begin{equation*}
		\int_{\C^N} \Phi \left( \dfrac{|F(z)|^2}{(1+|z|^2)^M}\right) \, d\nu(z),
	\end{equation*} 
	where 
    \begin{equation}\label{eq:polynomial from state vector}
    F(z) = \langle \otimes ^M(1,\overline{z}) | \psi  \rangle,\quad z \in \C^N
    \end{equation}
    is a holomorphic polynomial (since the inner product is linear in the second argument) on $\C^N$ of degree at most $M$ and $d\nu$ is a probability measure (see below for the explicit expression). Since the correspondence between unit pure state vectors and suitably normalized holomorphic polynomials of this kind is one-to-one, we are naturally led to consider the reproducing kernel Hilbert space $\P_M$ of holomorphic polynomials on $\C^N$ of degree at most $M$, endowed with the norm
	\begin{equation}\label{eq:norm on PM}
		\|F\|_{\P_M}^2 = \mathrm{dim}(\P_M) \int_{\C^N} \dfrac{|F(z)|^2}{(1+|z|^2)^M} \, d\nu(z),
	\end{equation}
	where 
    \begin{equation}\label{eq d}
    \mathrm{dim}(\P_M) = \binom{M+N}{N} \eqqcolon d
    \end{equation}
    and 
    \begin{equation}\label{eq dnu}
        d\nu(z) = c_N (1+|z|^2)^{-N-1}dA(z),\quad c_N = N!/\pi^N
    \end{equation}
    is a probability measure on $\C^N$; here $dA(z)$ denotes the Lebesgue measure on $\C^N$. The normalized monomials 
	\begin{equation*}
		e_{\alpha}(z) = c_{\alpha}z^{\alpha}, \quad z \in \C^N,\ \alpha \in \N^N,\ |\alpha| \leq M,
	\end{equation*}
	where $c_{\alpha} = \sqrt{\binom{M}{\alpha}}$, form an orthonormal basis of $\P_M$, therefore, we immediately have the expression of the reproducing kernel
	\begin{equation*}
		K(z,w) := \sum_{|\alpha|\leq M} e_{\alpha}(z) \overline{e_{\alpha}(w)} = (1+z \cdot \overline{w})^M,
	\end{equation*}
	whereas the normalized reproducing kernel at a point $w \in \C^N$ is
	\begin{equation*}
		k_w(z) := \dfrac{K(z,w)}{\|K(\cdot,w)\|_{\P_M}} = \dfrac{(1+z \cdot \overline{w})^M}{(1+|w|^2)^{M/2}}, \quad z \in \C^N.
	\end{equation*}
	With a slight abuse of notation, we define the Wehrl-type entropy of a normalized polynomial $F \in \P_M$ as
	\begin{equation}\label{eq:expression of sup}
		G(F) = \int_{\C^N} \Phi \left( \dfrac{|F(z)|^2}{(1+|z|^2)^M} \right) d\nu(z).
	\end{equation}
	Among all normalized polynomials in $\P_M$ this entropy is maximized by the normalized (multiples of the) reproducing kernels, that is, on the subset
	\begin{equation*}
		V \coloneqq \{e^{i\theta} k_{w}\, \colon \, w \in \C^N,\ \theta \in \R\} \subset \P_M,
	\end{equation*}
    since the normalized (multiples of the) reproducing kernels are exactly the polynomials that arise from coherent states through the formula \eqref{eq:polynomial from state vector}. Thus,  the supremum of $G$ over all possible normalized polynomials in $\mathcal{P}_M$ is achieved at any normalized reproducing kernel, and choosing  the reproducing 
    kernel at $w=0$ (that is, the constant function 1) we obtain
    \begin{equation*}
        \sup G = \int_{\C^N} \Phi \left( \dfrac{1}{(1+|z|^2)^M} \right) \, d\nu(z).
    \end{equation*}    
	Given a normalized polynomial $F \in \P_M$, the distance from $F$ to $V$ is, by definition,
	\begin{equation*}
		D[F] = \mathrm{inf}\{\|F-e^{i\theta}k_w\|_{\P_M} \, \colon \, w \in \C^N,\ \theta \in \R\}.
	\end{equation*}
	Then, Theorem \ref{th:stability in H_M} (for pure states) can be rephrased as follows.
	\begin{teo}\label{th:main_theorem_polynomials}
		For every convex function $\Phi \colon [0,1] \to \R$ that is strictly convex in some interval $(a,1)$ with $a \in (0,1)$, there exists a constant $c > 0$ such that
		\begin{equation*}
			\sup G - G(F) \geq c D[F]^2
		\end{equation*}
		for every normalized polynomial $F \in \P_M$.
	\end{teo}
   \begin{remark}
       Let $\rho = |\psi \rangle \langle \psi|$ and let $F \in \P_M$ be the polynomial associated with the pure state $\psi \in \mathcal{H}_M$. Since unit pure state vectors are in one-to-one correspondence with normalized polynomials in $\P_M$, and since we are dealing with finite-dimensional spaces, the distances $D[\rho]$ and $D[F]$ are equivalent, meaning that there is a constant $C>0$ such that
       \begin{equation*}
           C^{-1} D[\rho] \leq D[F] \leq C D[\rho].
       \end{equation*}
       For this reason, Theorem \ref{th:stability in H_M} (for pure states) and Theorem \ref{th:main_theorem_polynomials} are equivalent.
   \end{remark}
	To prove Theorem \ref{th:main_theorem_polynomials}, we will not work directly with polynomials, but rather with their coefficients in the monomial basis, expressing a 
    normalized polynomial $F \in \P_M$  as
	\begin{equation*}
		F(z) = \sum_{|\alpha| \leq M} X_{\alpha} e_{\alpha}(z), \quad z \in \C^N,
	\end{equation*}
	where the vector $X = (X_{\alpha})_{|\alpha| \leq M}$, due to the orthonormality
    of the $e_{\alpha}$'s, belongs to the $2d-1$ dimensional sphere $S \coloneqq S^{2d-1} \subset \mathbb{R}^{2d}\simeq\C^d$.  Hence, with a slight abuse of notation, for $X = (X_{\alpha}) \in S$ we can write  
	\begin{equation}\label{eq G(X)}
		G(X) = \int_{\C^N} \Phi \left( \dfrac{|\sum_{|\alpha| \leq M} X_{\alpha} e_{\alpha}(z)|^2}{(1+|z|^2)^M} \right) d\nu(z).
	\end{equation}
	In this way, we can regard the Wehrl-type entropy as a function $G \colon S \to \R$, which achieves its maximum at every point of the set
    \begin{equation*}
        V \coloneqq \Big\{X \in S \colon \sum_{|\alpha| \leq M} X_{\alpha} e_{\alpha} = e^{i\theta}k_{w}\ \text{for some } \theta \in \R, w \in \C^{N} \Big\}.
    \end{equation*}
    In the following, we will see that $V$ is a smooth submanifold of $S$, of dimension $2N+1$. 

Endowing $S$ with the Riemannian metric induced by $\R^{2d}\simeq\C^d$,
and denoting by $\mathrm{dist}(X,V)$ the geodesic distance from $X\in S$ to  $V$,
we can restate
 Theorem \ref{th:main_theorem_polynomials}  as follows.
	\begin{teo}\label{th:main theorem on the sphere}
		For every convex function $\Phi \colon [0,1] \to \R$ that is strictly convex in some interval $(a,1)$, with $a \in (0,1)$, there exists a constant $c > 0$ such that
		\begin{equation}\label{eq:main estimate on the sphere}
			\sup_{S} G - G(X) \geq c \, \mathrm{dist}(X,V)^2
		\end{equation}
		for every $X \in S$.
	\end{teo}
	This theorem is equivalent to Theorem \ref{th:main_theorem_polynomials} because $D[F]$ is
     the Euclidean distance from $X$ to $V$, and  the geodesic distance on $S$ is equivalent to the Euclidean distance, when $S$ is seen as a subset of $\C^d$. Indeed, if $X_1, X_2 \in S$,
    \[
    |X_1-X_2|=2\sin(\textrm{dist}(X_1,X_2)/2).
    \]
    \begin{remark}\label{remark: SU(N) invariance}        
        So far we have seen that there is a unitary correspondence between  $\mathcal{H}_{M}$ and the space of polynomials $\mathcal{P}_{M}$, and we now have another unitary map $\P_M\to\C^d \supset S$. Due to these correspondences, the action of $SU(N+1)$ on $\mathcal{H}_{M}$ is automatically transferred on $\mathcal{P}_M$, $\C^d$ and $S$. Moreover, since the set of coherent states in $\mathcal{H}_M$ is an orbit of $SU(N+1)$ and since the coherent states are in one-to-one correspondence with the points of $V$ (whether viewed as a subset of the normalized polynomials in $\mathcal{P}_M$ or as a subset of $S$), the isometric action of $SU(N+1)$ on $S$ is, in fact, transitive on $V$. This will be of great importance in the following, as it will allow us to transfer any result proved for a single point in $V$ to every other point in $V$.

        We point out that this unitary correspondence between $\mathcal{H}_M$ and $\C^d$ can be made explicit considering a particular orthonormal basis of $\mathcal{H}_M$. In fact, denoting with $\{\iota_n\}_{n=0}^N$  the canonical basis of $\C^{N+1}$, it is easy to see that the vectors
        \begin{equation*}
            \psi_{\alpha} = \sqrt{\binom{M}{\alpha}} P_M (\underbrace{\iota_0 \otimes \cdots \iota_0}_{M-|\alpha| \textrm{ times}} \otimes \underbrace{\iota_1 \otimes \cdots \iota_1}_{\alpha_1 \textrm{ times}} \otimes \cdots \otimes \underbrace{\iota_N \otimes \cdots \otimes \iota_N}_{\alpha_N \textrm{ times}}),
        \end{equation*}
        with $\alpha \in \N^N$ of length at most $M$, form an orthonormal basis of $\mathcal{H}_M$ and are such that
        \begin{equation*}
            \langle \otimes^M (1,\overline{z}) | \psi_{\alpha} \rangle = e_{\alpha}(z).
        \end{equation*}
        Therefore, the unitary correspondence that we obtained through $\mathcal{P}_M$ is simply the map
        \begin{equation*}
            \mathcal{H}_M \ni \psi = \sum_{|\alpha| \leq M} X_{\alpha} \psi_{\alpha} \mapsto (X_{\alpha})_{|\alpha| \leq M} \in \C^d.
        \end{equation*}
    \end{remark}
    
	\subsection{Preliminaries of differential geometry}\label{subsec:differential geometry}
	We conclude this section with some preliminaries of differential geometry. First of all, we notice that the subset $V$ where $G$ achieves its maximum is a smooth submanifold of $S$ of dimension $2N+1$ since it is the preimage of the Veronese variety under the canonical projection $\pi \colon S \to \C P^d$ (see \cite[][Section 11.3]{fulton_harris} and \cite[Remark 2.1 (ii)]{NRT_wehrl_SU(N)}), which is smooth (see \cite[][Page 184]{harris_algebraic_geometry}), and the above projection is a submersion (see \cite[][Theorem 6.30]{lee_smooth_manifolds}). For more information on the connection between coherent states and rational curves, see, for example, \cite{brody_graefe}.
    
    For any point $X \in V \subset S$ we have the decomposition 
    \[
    T_XS = T_XV \oplus (T_XV)^{\perp}.
    \]
    In the following, we are going to consider the $\varepsilon$-neighborhood of $V$, that is the set \[
    U_{\varepsilon} = \{X \in S \colon \mathrm{dist}(X,V) < \varepsilon\},
    \]
    where the distance is intended in the geodesic sense. Then, it is well known (see \cite[Theorem 6.40]{lee_introduction}) that $U_{\varepsilon}$ is also a $\varepsilon$-tubular neighborhood of $V$, which means that every point in $U_\varepsilon$ is of the form $X(t)$ where $X \colon (-\varepsilon,\varepsilon) \to S$ is a geodesic with unit speed and such that $X(0) \in V$, $X'(0) \in (T_{X(0)}V)^{\perp} \subset T_{X(0)}S$. In other words, $U_{\varepsilon}$ is the image of the exponential map restricted to the normal bundle of $V$.

    Rather than at a generic point $X\in S$, we will perform the computations 
    at
    $X_0 = (1,0,\ldots,0)$, where the expressions of the tangent and normal spaces take
    the simple forms
    \begin{equation*}
        \begin{aligned}
             T_{X_0}S &= \{Y = (Y_{\alpha})_{|\alpha| \leq M} \in \C^d \colon \mathrm{Re}Y_0=0\}, \\ 
             T_{X_0}V &= \{Y = (Y_{\alpha})_{|\alpha| \leq M} \in T_{X_0}S \colon Y_{\alpha}=0 \text{ for } |\alpha| \geq 2\}, \\
             (T_{X_0}V)^{\perp} &= \{Y = (Y_{\alpha})_{|\alpha| \leq M} \in T_{X_0}S \colon Y_{\alpha}=0 \text{ for } |\alpha| \leq 1\}.
        \end{aligned}       
    \end{equation*}
    To obtain  the expression for $T_{X_0}V$,  consider a smooth curve $v \colon (-\varepsilon,\varepsilon) \to S$ such that $v(0) = v_0 = (1,0,\ldots,0) \in \C^{N+1}$, and note that  $\otimes^M v(t)$ is a curve on the submanifold of coherent states passing through $\otimes^M v_0$, whose corresponding polynomials are given by
    \begin{equation*}
        \langle \otimes^M(1,\overline{z}) | \otimes^M v(t)  \rangle = \langle (1,\overline{z}) | v(t) \rangle^M = \sum_{|\alpha| \leq M} X_{\alpha}(t) e_{\alpha}(z), \quad z \in \C^N
    \end{equation*}
    where $X(t)$ is a smooth curve in $V \subset S$ such that $X(0) = X_0$.  
    Differentiation at  $t=0$ gives
    \begin{equation*}
        M(v_0'(t) + \sum_{j=1}^{N} v_j'(t) z_j )= \sum_{|\alpha| \leq M} X_{\alpha}'(0) e_{\alpha}(z),
    \end{equation*}
    which proves the inclusion  $T_{X_0} V\subseteq \{Y \in T_{X_0}S   \colon Y_{\alpha}=0 \text{ for } |\alpha| \geq 2\}$. The reverse inclusion follows observing that
the two spaces have the same dimension, since $V$ is a submanifold of $S$ of dimension $2N+1$.

	\section{Proof of Theorem \ref{th:main theorem on the sphere} when $\Phi$ is of class $C^2$}\label{sec:proof regular case}
	
	To illustrate the main idea of the proof, in this section we prove Theorem \ref{th:main theorem on the sphere} under the additional assumption that $\Phi\in C^2([0,1])$. In this case, also the Wehrl-type entropy function $G$ in \eqref{eq G(X)} is easily seen to be of class $C^2$, which simplifies the proof. The key points of the proof are:
	\begin{itemize}
		\item the value of $G$ at points far from $V$ is uniformly distant from the maximal value of $G$;
		\item the second differential of $G$ at every point $X \in V$ is negative semidefinite on $T_X S$ and negative definite on the orthogonal space $(T_XV)^{\perp}$.
	\end{itemize}
    These ideas are formalized by the following lemmas, the first of which does not actually depend on the regularity of $G$ and, in fact, will also be used in the proof for general $\Phi$ in  Section \ref{sec:proof non regular case}. We also observe that the constant $c$ in this lemma is explicit,  in the sense that it is not
obtained by a compactness argument.
	\begin{lem}\label{lem:estimate for points far from V}
		Let $\Phi \colon [0,1] \to \R$ be a convex function that is strictly convex in some interval $(a,1)$ with $a \in (0,1)$. Then, for every $\varepsilon>0$ small enough there exists $c>0$ such that, for every $X \in S$, if $\mathrm{dist}(X,V) \geq \varepsilon$ then 
        \begin{equation*}
            \sup_{S} G - G(X) \geq c.
        \end{equation*}
    In particular, up to reducing the constant $c$, we have
    \begin{equation*}
        \sup_{S} G - G(X) \geq c\,\mathrm{dist}(X,V)^2\quad \textrm{whenever}\quad \mathrm{dist}(X,V) \geq \varepsilon.
    \end{equation*}
	\end{lem}
    \begin{proof}
           We use the notation of Section \ref{sec:notation and setting}.  Let $F = \sum_{|\alpha| \leq M} X_{\alpha} e_{\alpha} \in \P_M$ be the polynomial associated with $X=(X_\alpha) \in S$ and let
            \begin{equation*}
                T := \sup_{z \in \C^N} \dfrac{|F(z)|^2}{(1+|z|^2)^M} \in [0,1].
            \end{equation*}
            We start by proving that, if $\mathrm{dist}(X,V) \geq \varepsilon$, then $1-T \geq \tau$ for some positive constant $\tau$ that depends only on $\varepsilon$. We have
            \begin{align*}
                D[F]^2 &= \inf_{\theta \in \R, w \in \C^N} ||F-e^{i\theta}k_w||_{\mathcal{P}_M}^2 = \inf_{\theta \in \R, w \in \C^N} \left(2-2\mathrm{Re}\, e^{-i\theta} \dfrac{F(w)}{(1+|w|^2)^{M/2}}\right)\\
                       &= 2-2\sup_{\theta \in \R, w \in \C^N} \mathrm{Re}\, e^{-i\theta} \dfrac{F(w)}{(1+|w|^2)^{M/2}} = 2-2 \sup_{w \in \C^N} \dfrac{|F(w)|}{(1+|w|^2)^{M/2}} = 2-2\sqrt{T}.
            \end{align*}
            As observed in Section \ref{sec:notation and setting}, $D[F]$ is just the Euclidean distance from $X$ to $V$ in $\C^d$, while
 for the geodesic distance in the unit sphere we have $\mathrm{dist}(X,V) \leq C D[F]$ for some positive constant $C$. Therefore 
            \begin{equation*}
                \varepsilon \leq \mathrm{dist}(X,V) \leq C D[F] =  C\sqrt{2-2\sqrt{T}},
            \end{equation*}
            and rearranging this expression we obtain that $1-T \geq \tau$ for every $\varepsilon$ small enough and for some $\tau>0$ depending on $\varepsilon$.
    
             Let $F$, $X$ and $T$ be as above, with ${\rm dist}(X,V)\geq \varepsilon$
            so that $1-T \geq \tau$, and consider the decomposition $\Phi = \Phi_1 + \Phi_2$, where
            \begin{equation*}
                \Phi_1(t) = \begin{cases}
                    \Phi(t), \quad &t \in [0,T] \\
                    \Phi'_{-}(T)(t-T)+\Phi(T), \quad & t \in (T,1]
                \end{cases}
            \end{equation*}
            ($\Phi'_{-}$ denotes the left derivative of $\Phi$).
            Denoting by $G_1$ and $G_2$ the corresponding Wehrl-type entropies, since both $\Phi_1$ and $\Phi_2$ are convex, we have
            \begin{equation*}
                \sup_S G - G(F) = \sup_S G_1 - G_1(F) + \sup_S G_2 - G_2(F) \geq \sup_S G_2 - G_2(F) = \sup_S G_2,
            \end{equation*}
            where the first equality holds since $G_1$ and $G_2$ both achieve their maximum on $V$ (by the Lieb--Solovej inequality) while the last equality follows from the fact that $\Phi_2(t) = 0$ for $t \leq T$.
            
             To give an estimate from below of $\sup G_2$, we consider the distribution function of $1/(1+|z|^2)^M$, that is,
            \[
            \mu_0(t)=\nu(\{z\in\C^N:\ 1/(1+|z|^2)^M>t\}).
            \]
            Its explicit expression $\mu_0(t) := (1-t^{1/M})^{N}$ for $t\in [0,1]$ (see \cite[Equation (4.3)]{NRT_wehrl_SU(N)}) is not needed here, 
            where we just observe that $\mu_0(t)>0$ for $t\in [0,1)$. 
            
            Recalling the expression of the supremum of $G_2$ \eqref{eq:expression of sup} and using the layer cake representation \cite[][Theorem 1.13]{liebloss}, since $\lim_{t\to 0^+}\Phi_2(t)=0$ we have
            \begin{align*}
                \sup G_2 &= \int_{\C} \Phi_2 \left( \dfrac{1}{(1+|z|^2)^M} \right) \, d\nu(z)\\
                         & =  \int_0^1 \Phi'_2(t) \mu_0(t) \, dt \\
                         &= \int_T^1 (\Phi'(t) - \Phi'_{-}(T)) \, \mu_0(t) \, dt \\
                         & \geq \int_{1-\tau}^1 (\Phi'(t) - \Phi'_{-}(1-\tau)) \, \mu_0(t) \, dt \\
                         & = c \geq \dfrac{c}{2}\cdot 2(1-\sqrt{T}) = \frac{c}{2} D[F]^2 \geq c' \, \mathrm{dist}(X,V)^2,
            \end{align*}
where $c,c'>0$ because by assumption $\Phi$ is strictly convex in an interval of positive length contained in $(1-\tau,1)$.
             This concludes the proof. 
    \end{proof}
We now establish the desired stability estimate near $V$, for which the following lemma
plays a crucial role (we alert the reader that some of the computations are deferred to the appendix to maintain continuity of reading).
    \begin{lem}\label{lem:second differential is negative definite}
        Let $\Phi \colon [0,1] \to \R$ be a convex function of class $C^2([0,1])$ that is strictly convex in some interval $(a,1)$ with $a \in (0,1)$. Then, the second differential of $G$ at $X_0 \in V$ is negative  semi-definite on $T_{X_0}S$ and negative definite on the orthogonal space $(T_{X_0}V)^{\perp}$, that is
        \begin{itemize}
            \item $\mathrm{d}^2_{X_0}G(Y) \leq 0 $ for every $Y \in T_{X_0} S$;
            \vspace{2mm}
            \item $\mathrm{d}^2_{X_0}G(Y) \leq - c|Y|^2$ for every $Y \in (T_{X_0}V)^{\perp}$,
        \end{itemize}
        where $c>0$.
    \end{lem}
    \begin{proof}
        By $SU(N+1)$ invariance (see Remark \ref{remark: SU(N) invariance}), it suffices to consider the point $X_0 = (1,0,\ldots,0) \in V$. Consider a smooth curve $X \colon (-\varepsilon,\varepsilon) \to S$ such that $X(0) = X_0$ and let $Y = (Y_{\alpha})_{|\alpha| \leq M} = X'(0) \in T_{X_0}S$.  Since $X(0) = (1,0,\ldots,0)$ and $\mathrm{Re}Y_0 = 0$ (see Section \ref{subsec:differential geometry}),
        \begin{align*}
            \mathrm{d}^2_{X_0}G(Y) :=& \frac{d^2}{dt^2}G(X(t))|_{t=0} =\int_{\C^N} \Phi''\left(\dfrac{1}{(1+|z|^2)^M}\right) \dfrac{(2\mathrm{Re} \sum_{\alpha \neq 0} Y_{\alpha}e_{\alpha}(z))^2}{(1+|z|^2)^{2M}} d\nu(z)\\
                                                  +&\!\int_{\C^N} \!\Phi' \! \left(\!\dfrac{1}{(1+|z|^2)^M}\!\!\right) \! \dfrac{2\mathrm{Re} \sum_{|\alpha| \leq M} X''_{\alpha}(0)e_{\alpha}(z) + 2 |\sum_{|\alpha| \leq M} Y_{\alpha} e_{\alpha}(z)|^2}{(1+|z|^2)^{M}} \, d\nu(z).
        \end{align*}
        We point out that the term $\mathrm{Re} \sum_{|\alpha| \leq M} X''_{\alpha}(0)e_{\alpha}(z)$ can be replaced with just $\mathrm{Re}X_0''(0)$ (the pedex 0 stands for the multi-index in $\N^N$ with 0 in all components). For, passing to polar coordinates and integrating in 
        the angular component first, one sees that the integrals of $\mathrm{Re} \, e_{\alpha}(z)$ vanish for $\alpha \neq 0$, since $e_{\alpha}(0)=0$ and $e_\alpha(z)$  (being the real
        part of a holomorphic function) is harmonic. We also observe that 
        \[
        \mathrm{Re}X_0''(0) = -\sum_{|\alpha| \leq M} |Y_{\alpha}|^2, 
        \]
        as can be seen by differentiating twice the identity $|X(t)|^2=1$ and evaluating at $t=0$.
        Concerning the other terms, we have
    \begin{align*}
    (2\mathrm{Re} \sum_{\alpha \neq 0} Y_{\alpha}e_{\alpha}(z))^2 &= 2\sum_{\alpha \neq 0} |Y_{\alpha}|^2 |e_{\alpha}(z)|^2 + 2 \mathrm{Re} [\sum_{\alpha,\beta\not=0} Y_{\alpha} Y_{\beta} e_{\alpha}(z) e_{\beta}(z)\\
    &+\sum_{\alpha \neq \beta; \alpha,\beta\not=0} Y_{\alpha}\overline{Y_{\beta}} e_{\alpha}(z) \overline{e_{\beta}(z)}].
     \end{align*}
Again, using polar coordinates one sees that all the terms on the right-hand side, with the
exception of the first, give a zero contribution in the integral (the term $\textrm{Re}\, Y_\alpha \overline{Y_\beta}e_\alpha(z) \overline{e_\beta(z)}$ is not harmonic in general, 
yet its integral on every shell is zero
since $\alpha\not=\beta$). Similarly,
    \[
|\sum_{|\alpha| \leq M} Y_{\alpha}e_{\alpha}(z)|^2 = \sum_{|\alpha| \leq M} |Y_{\alpha}|^2 |e_{\alpha}(z)|^2 + \sum_{\alpha \neq \beta} Y_{\alpha} \overline{Y_{\beta}} e_{\alpha}(z)\overline{e_{\beta}(z)}
\]
and, again, the integral of the last sum is zero.
  
   Thus, since $|e_0(z)|^2=1$ we are left with
  \begin{equation}\label{eq:second differential in compact form}
      \dfrac{1}{2} \mathrm{d}^2_{X_0}G(Y) = \sum_{\alpha \neq 0} b_{\alpha} |Y_{\alpha}|^2
  \end{equation}
  where
  \begin{align*}
            b_{\alpha} = \int_{\C^N} \Phi''\left(\dfrac{1}{(1+|z|^2)^M}\right)& \dfrac{|e_{\alpha}(z)|^2}{(1+|z|^2)^{2M}} \, d\nu(z)\\
            &+\int_{\C^N} \Phi'\left(\dfrac{1}{(1+|z|^2)^M}\right) \dfrac{|e_{\alpha}(z)|^2-1}{(1+|z|^2)^{M}} \, d\nu(z).
        \end{align*}
        So, to prove the lemma, we just need to prove that every coefficient $b_{\alpha}$ is nonnegative for $|\alpha|=1$ and strictly negative for $2 \leq |\alpha| \leq M$. To do so, we pass to the polar coordinates $z = \rho \omega$, where $\rho\geq0$ and  $\omega \in S^{2N-1} \subset \C^N$ (unit sphere in $\C^N$), so that
        \[
        d\nu=\frac{c_N\rho^{2N-1}}{(1+\rho^2)^{N+1}}\, d\rho\,d\omega,
        \]
        where $c_N$ is the normalization constant for the measure $d\nu$ in \eqref{eq dnu}. 
        Observe that $1/c_N$ is the volume of the unit ball in $\C^N$ (that is, the real $2N$-ball), therefore $c_N/(2N) d\omega$ is the normalized surface measure on $S^{2N-1}$. 
        For ease of notation we let 
        \begin{equation}\label{eq:definition c alfa tilde}
            \tilde{c}_{\alpha}^2 \coloneqq c_{\alpha}^2 \int_{S^{2N-1}} |\omega^{\alpha}|
            ^2 \dfrac{c_N}{2N} \, d\omega,
        \end{equation}
        whose explicit expression can be found below in Lemma \ref{lemma:expression of c alfa tilde}. 
        
        Hence, passing to polar coordinates in the integrals defining $b_{\alpha}$, 
        after the further change of variable $\rho^2=s$ we have to study the sign of 
        \begin{equation*}
            N\!\int_0^{\infty} \Phi''\left( \! \dfrac{1}{(1+s)^M} \!\right) \! \dfrac{\tilde{c}_{\alpha}^2 s^{|\alpha|+N-1}}{(1+s)^{2M+N+1}}ds + N\!\int_0^{\infty} \Phi'\left(\!\dfrac{1}{(1+s)^M}\!\right) \!\dfrac{\tilde{c}_{\alpha}^2 s^{|\alpha|+N-1}-s^{N-1}}{(1+s)^{M+N+1}} \, ds.
        \end{equation*}
        In view of integrating by parts the second term, 
        using Lemma \ref{lem:primitive} we can preliminarily  compute
        \begin{align*}
            \int_0^s \dfrac{\tilde{c}_{\alpha}^2 \sigma^{|\alpha|+N-1}-\sigma^{N-1}}{(1+\sigma)^{M+N+1}} \, d\sigma &= \dfrac{\tilde{c}_{\alpha}^2 (1+s)^{-(M+N)}}{A_{M,N,|\alpha|}} \sum_{j=|\alpha|+N}^{M+N}\binom{M+N}{j}s^j \\
            &- \dfrac{(1+s)^{-(M+N)}}{A_{M,N,0}} \sum_{j=N}^{M+N} \binom{M+N}{j}s^j \\
            &= - \dfrac{\tilde{c}_{\alpha}^2 (1+s)^{-(M+N)}}{A_{M,N,|\alpha|}} \sum_{j=N}^{|\alpha|+N-1} \binom{M+N}{j}s^j,
        \end{align*}
        where in the last step we used Lemma \ref{lemma:expression of c alfa tilde} (see below for the explicit expression of $A_{M,N,|\alpha|}$). From this and from the fact that $\Phi'(t)$ has finite limits at  $0^+$ and at $1^-$, it is immediate to see that the boundary term in the integration by parts vanishes and therefore we end up with      
        \begin{equation}\label{eq:expression of b_alpha}
            \begin{aligned}
               b_{\alpha} = N\int_0^{\infty} \Phi''\left(\dfrac{1}{(1+s)^M}\right)&\dfrac{\tilde{c}_{\alpha}^2}{(1+s)^{2M+N+1}} \Big[s^{|\alpha|+N-1} \\
               -&\dfrac{M}{A_{M,N,|\alpha|}} \sum_{j=N}^{|\alpha|+N-1} \binom{M+N}{j}s^j\Big] \, ds.
           \end{aligned}
        \end{equation}
                Now, we have \[
       \frac{M}{A_{M,N,1}}\binom{M+N}{N}=1,
       \]
       so for $|\alpha|=1$ the integral is 0 (which is clear since  $Y\in T_{X_0}V$ and $G$ is constant on $V$), while for $|\alpha| > 1$ we have
       \[
       \displaystyle \dfrac{M}{A_{M,N,|\alpha|}}\binom{M+N}{|\alpha|+N-1}=\dfrac{M}{M-|\alpha|+1} > 1,
       \]
       that is, the coefficient of the monomial of degree $|\alpha|+N-1$ is strictly negative. Hence, the expression in square brackets is strictly negative for every $s>0$. Finally, by assumption, $\Phi''\geq0$ and $\Phi$ is strictly convex in an interval of the type $(a,1)$, for some $a\in (0,1)$, and the desired conclusion follows. 
    \end{proof}
    
    \begin{proof}[Proof of Theorem \ref{th:main theorem on the sphere} for $\Phi$ of class $C^2$]
		Assume, in addition to the hypothesis in the statement, that $\Phi\in C^2([0,1])$. Let $\varepsilon>0$ to be fixed later. As already recalled in Section \ref{subsec:differential geometry}, every point in the $\varepsilon$-neighborhood of $V$ can be reached by a geodesic $X \colon (-\varepsilon,\varepsilon) \to S$ parametrized with unit speed such that $X(0) = X_0 \in V$, with an initial velocity $Y=X'(0) \in (T_{X_0}V)^{\perp}$ such that $|Y|=1$. In this case, for $\varepsilon$ sufficiently small and $t\in (-\varepsilon,\varepsilon)$ we have, for some $\xi\in [-|t|,|t|]$ and $c>0$,
		\begin{align*}
			\sup_{S} G-G(X(t)) &= G(X(0))-G(X(t)) \\
			&=-\dfrac{1}{2}(G \circ X)''(\xi )t^2 \\
			&\geq c t^2 = c\ {\rm dist} (X(t),V)^2,
		\end{align*}
		where we applied Lemma \ref{lem:second differential is negative definite} and the fact that the function $(G \circ X)''(t)$ is continuous, uniformly with respect to $Y\in (T_{X_0}V)^{\perp}$, with $|Y|=1$. This fixes the value of $\varepsilon$ and gives the desired estimate where ${\rm dist}(X,S)<\varepsilon$. Then, for this value of $\varepsilon$ we apply Lemma \ref{lem:estimate for points far from V} and obtain the  estimate also where $\mathrm{dist}(X,V) \geq \varepsilon$.
	\end{proof}
    
    \begin{remark}\label{remark: xlogx}
        Wehrl's original entropy is defined with $\Phi(t) = t \log t$, which is not in $C^2([0,1])$. Nevertheless, one can easily adapt the above proof to this case, using the decomposition $\Phi = \Phi_1+\Phi_2$ where
        \begin{equation*}
            \Phi_1(t) = \begin{cases}
                \dfrac{1}{2}\Phi''(a)(t-a)^2+\Phi'(a)(t-a)+\Phi(a), &t \in [0,a] \vspace{2mm}\\ 
                \Phi(t), \quad &t \in (a,1]
            \end{cases},
        \end{equation*}
        for some $a \in (0,1)$. Indeed, since $\Phi_2:=\Phi-\Phi_1$ is convex, reasoning as
        in the proof of Lemma \ref{lem:estimate for points far from V} one sees
        that it is enough to prove the desired estimate for the Wehrl-type entropy associated with $\Phi_1$, which belongs to $C^2([0,1])$ and satisfies the assumptions of Theorem \ref{th:main theorem on the sphere}.
    \end{remark}
    
	\section{Proof Theorem \ref{th:main theorem on the sphere} for an arbitrary convex function}\label{sec:proof non regular case}
	
	When $\Phi$ is not $C^2$ the previous argument breaks down and some additional work is required.
    In fact, we will approximate $\Phi$ by smooth functions and compute the second differential of the smoothed  $G$'s at points that are not in $V$, in which case the computation is more involved. The proof of the following lemma is postponed to Appendix \ref{sec:appendix}. 
     \begin{lem}\label{lem:second differential non regular case}
       Let $X_0=(1,0,\ldots,0)\in\C^d$. For every $0<a<b<1$ there exist $\varepsilon>0$ and a continuous function $h(\tau,t,Y)$, defined for $\tau\in (a,b)$, $|t|<\varepsilon$ and $Y\in (T_{X_0} V)^\perp$ with $|Y|=1$, such that the following representation  formula holds. 
        
        If $X(t,Y)$ denotes the unit-speed geodesic in $S\subset\C^d$, with $X(0,Y)=X_0$ and initial velocity $\partial_t X(0,Y)=Y$, 
         for every convex function $\Phi \in C^2([0,1])$ such that $\Phi''$ is supported in $(a,b)$  it holds \begin{equation}\label{eq:second differential at arbitrary point}
			\partial_{tt}^2 G(X(t,Y))= \int_0^1 \Phi''(\tau) h(\tau,t,Y) \, d\tau,\quad |t|<\varepsilon,\  Y\in (T_{X_0} V)^\perp,\ |Y|=1.
		\end{equation}
        Moreover, $h(\tau,0,Y)<0$ for every $\tau\in (a,b)$ and $Y\in (T_{X_0} V)^\perp$ with $|Y|=1$. 
    \end{lem}	
	\begin{proof}[Proof of Theorem \ref{th:main theorem on the sphere} in full generality]
    Let $\varepsilon>0$ be sufficiently small, to be fixed later. In view of Lemma \ref{lem:estimate for points far from V} it is sufficient to prove the estimate 
    \begin{equation}\label{eq due}
    \sup_S G-G(X)\geq c\,{\rm dist}(X,V)^2\qquad \textrm{if } {\rm dist}(X,V)<\varepsilon.
    \end{equation}
    To this end, let $0<a<b<1$ and decompose $\Phi = \Phi_1 + \Phi_2$, where $\Phi_1$ is the affine continuation of $\Phi \vert_{[a,b]}$, that is
        \begin{equation*}
            \Phi_1(t) = \begin{cases}
                \Phi'_{+}(a)(t-a)+\Phi(a), &t \in [0,a)\\
                \Phi(t), &t \in [a,b]\\
                \Phi'_{-}(b)(t-b) + \Phi(b), &t \in (b,1].
            \end{cases}
        \end{equation*}
        Then, since both $\Phi_1$ and $\Phi_2$ are convex, denoting with $G_1$ and $G_2$ the corresponding entropies we have
        \begin{align*}
            \sup_S G - G(X) = \sup_S G_1 - G_1(X) + \sup_S G_2 - G_2(X) \geq \sup_S G_1 - G_1(X),
        \end{align*}
        where in the first equality we used the fact that both $G_1$ and $G_2$ achieve their supremum at every point of $V$ (by Lieb--Solovej inequality). Therefore, it suffices to prove the estimate \eqref{eq due} for $\Phi_1$. Also, by the assumption on $\Phi$, if $b$ is sufficienly close to $1$, $\Phi_1$ will be strictly convex in some small interval contained in $[a,b]$. 
        
        Hence, we can assume that $\Phi:[0,1]\to\mathbb{R}$ is convex, affine in $[0,a]$ and in $[b,1]$, with $\Phi''$, regarded as a positive Radon measure in (0,1), non zero and supported in $[a,b]$. 
		
		Again we recall that every point of the $\varepsilon$-neighborhood of $V$ in $S$ is of the form $X(t)$, where $X \colon (-\varepsilon,\varepsilon) \to S$ is a geodesic, parametrized with unit speed,  starting from a point $X_0 \in V$ and such that $Y = X'(0) \in (T_{X_0}V)^{\perp}$. With this in mind, the inequality \eqref{eq due} can be rephrased as 
		\begin{equation}\label{eq uno}
			G(X(0)) - G(X(t)) \geq c t^2 \quad \forall t \in (-\varepsilon,\varepsilon)
		\end{equation}
		for some $c>0$, since $|t|$ is exactly the geodesic  distance from $X(t)$ to $V$. 
        
        We now approximate $\Phi$ with more regular convex functions. First, we extend $\Phi$ to a function defined on $\mathbb{R}$, affine on $(-\infty,a]$ and on $[b,+\infty)$. Then, by convolution with a standard mollifier (see \cite[][Section 6.3]{evans_gariepy}) we construct a sequence of convex functions $\Phi_n:[0,1]\to\R$ such that, possibly for a slightly larger interval $[a,b]\subset (0,1)$,
		\begin{itemize}
			\item $\Phi_n\in C^2([0,1])$;
			\item $\Phi_n \to \Phi$ uniformly in $[0,1]$;
			\item  $\Phi_n''$ has compact support in  $(a,b)$;
			\item $\Phi_n'' \to \Phi$ as (positive, finite) Radon measures on (0,1), that is, in the distribution sense.
		\end{itemize}
		We denote with $G_n$ the corresponding Wehrl-type entropies. Using dominated convergence we have
		\begin{equation*}
			G(X(0)) - G(X(t)) = \lim_{n \to \infty} \big(G_n(X(0)) - G_n(X(t))\big)
		\end{equation*}
		and applying Taylor's formula to the $G_n$'s we see that
		\begin{equation*}
			G_n(X(0)) - G_n(X(t)) = -\dfrac{1}{2}(G_n\circ X)''(\xi_n) t^2
		\end{equation*}
		for some $\xi_n$ such that $|\xi_n| \leq|t|$. Therefore, to prove the stability estimate \eqref{eq uno}, we just need to prove that $(G_n \circ X)''(\xi) < -c < 0$ for some constant $c$ independent of $n$ and for every $\xi \in (-\varepsilon,\varepsilon)$. By $SU(N+1)$ invariance (see Remark \ref{remark: SU(N) invariance}) it is sufficient to prove this when $X(0)=X_0 = (1,0,\ldots,0) \in \C^d$.
		
		From Lemma \ref{lem:second differential non regular case} (applied to a slightly larger interval $(a,b)$) we have
		\begin{equation*}
			(G_n \circ X)''(\xi) = \int_0^1 \Phi_n''(\tau) h(\tau, \xi,Y) d\tau,
		\end{equation*}
		(recall, $Y=X'(0)$) where $h(\tau,0,Y) < -c < 0$ for $\tau \in [a,b]$ and $Y\in (T_{X_0}V)^\perp$ with $|Y|=1$. However, since $h$ is continuous, possibly with a smaller constant $c$ we have, for $\varepsilon$ sufficiently small, that $h(\tau,\xi,Y) < -c < 0$ for every $\tau \in (a,b)$, $\xi \in (-\varepsilon,\varepsilon)$ and $Y\in (T_{X_0}V)^\perp$ with $|Y|=1$. Hence,  for $t\in(-\varepsilon,\varepsilon)$ we have 
		\begin{align*}
			G(X(0)) - G(X(t)) &= \lim_{n \to \infty} G_n(X(0)) - G_n(X(t)) \\
			&= \lim_{n \to \infty} -\dfrac{1}{2}(G_n \circ X)''(\xi_n)t^2\\
			&= \lim_{n \to \infty} -\dfrac{t^2}{2}\int_a^b \Phi_n''(\tau) h(\xi_n,\tau) d\tau,
            \end{align*}
            and therefore
            \begin{align*}
           G(X(0)) - G(X(t))
			&\geq c t^2\lim_{n \to \infty} \int_a^b \Phi_n''(\tau) d\tau \\
			&= c t^2\lim_{n \to \infty} \int_0^1 \Phi_n''(\tau) d\tau\\
			&= c\Phi''((0,1)) t^2 =  c' t^2
		\end{align*}       
		where $\Phi''((0,1))>0$ denotes the total mass of the measure $\Phi''$. The stability estimate \eqref{eq uno} is therefore proved and this concludes the proof. 
	\end{proof}
	
	\appendix
	\section{Technical lemmas}\label{sec:appendix}
    In this appendix, we have gathered several technical lemmas. 
	\begin{lem}\label{lem:primitive}
		Let $N \geq 1$, $M \geq 1$ and $0 \leq K \leq M$. Then, for $s\geq0$, 
        \begin{equation}\label{eq:primitive}
			\int_0^s \dfrac{\sigma^{K+N-1}}{(1+\sigma)^{M+N+1}} \, d\sigma = \dfrac{(1+s)^{-(M+N)}}{A_{M,N,K}}\sum_{j=K+N}^{M+N}\binom{M+N}{j}s^j
		\end{equation}
		where
		\begin{equation}\label{eq:constant in the primitive}
			A_{M,N,K} = (M-K+1) \binom{M+N}{K+N-1}.
		\end{equation}
	\end{lem}
	\begin{proof}
	    After changing the variable $\tau = \sigma/(1+\sigma)$ one recognizes the expression of the incomplete beta function (up to the normalization constant), for which the expression \eqref{eq:primitive} has been known for a long time (see \cite[][Formula 2.5]{dutka_incomplete_beta_functions}).
	\end{proof}
	\begin{lem}\label{lemma:expression of c alfa tilde}
        The coefficients $\tilde{c}_{\alpha}^2$ defined in \eqref{eq:definition c alfa tilde} are given by 
		\begin{equation*}
			\tilde{c}_{\alpha}^2 = \dfrac{(N-1)!M!}{(N+|\alpha|-1)!(M-|\alpha|)!} = \dfrac{A_{M,N,|\alpha|}}{A_{M,N,0}},
		\end{equation*}        
        where the numbers $A_{M,N,|\alpha|}$ are defined in \eqref{eq:constant in the primitive}.
	\end{lem}
    \begin{proof}
        Recalling the expression of the norm \eqref{eq:norm on PM} in $\P_M$ we have
        \begin{align*}
            1 &= d \int_{\C^N} \dfrac{c_{\alpha}^2|z^{\alpha}|^2}{(1+|z|^2)^{M+N+1}} c_N dA(z) \\
              &= dN \int_{S^{2N-1}} c_{\alpha}^2 |\omega^{\alpha}| \dfrac{c_N}{2N} d\omega \int_0^{\infty} \dfrac{s^{|\alpha|+N-1}}{(1+s)^{M+N+1}} ds \\
              &= \dfrac{dN}{A_{M,N,|\alpha|}}\tilde{c}_{\alpha}^2,
        \end{align*}
        where in the last equality we used \eqref{eq:primitive}. Then, recalling the explicit expression of $d$ \eqref{eq d} it is immediately seen that $dN = A_{M,N,0}$.
    \end{proof}
	\begin{proof}[Proof of Lemma \ref{lem:second differential non regular case}]
    For simplicity of notation, let us set $X(t)=X(t,Y)$. Similarly, we will omit the dependence on $Y$ of various functions below.
   
       Let $\varepsilon>0$, to be chosen later on.  Let $\Phi$ be as in statement, in particular with  $\Phi''$ compactly supported in $(a,b)$.   We have to compute the second derivative of
		\begin{equation*}
			(G \circ X)(t) = \int_{\C^N} \Phi(u(z,t)) \, d\nu(z)
		\end{equation*}
		for $|t|<\varepsilon$, where
		\begin{equation*}
			u(z,t)=u(z,t,Y)=\dfrac{|\sum_{|\alpha| \leq M} X_{\alpha}(t) e_{\alpha}(z) |^2}{(1+|z|^2)^M}
		\end{equation*}
		(recall that $X_\alpha(t)=X_\alpha(t,Y)$ will also depend smoothly on $Y$).

        Since $\int_{\C^N} u(z,t) d\nu(z) = 1$ for all $t$ and $Y$, we can suppose, without loss of generality, that $\Phi(0) = \Phi'(0) = 0$. Therefore, we have
        \begin{equation*}
            \Phi(\upsilon) = \int_0^{\upsilon} \int_0^{\sigma} \Phi''(\tau) \, d\tau \, d\sigma,
        \end{equation*}
        which allows us to rewrite
        \begin{equation*}
            (G \circ X)(t) = \int_{\C^N} \int_0^{u(z,t)} \int_0^ \sigma \Phi''(\tau) \, d\tau \, d\sigma \, d\nu(z).
        \end{equation*}
        Differentiating twice with respect to $t$ leads to
        \begin{align*}
            (G \circ X)''(t) &= \int_{\C^N} \Phi''(u(z,t)) (\partial_tu(z,t))^2 \, d\nu(z)\\
            &\qquad\qquad\qquad\qquad+ \int_{\C^N} \Big(\int_0^{u(z,t)} \Phi''(\tau) \, d\tau \Big) \partial^2_{tt} u(z,t) \, d\nu(z) \\
            &= \int_{\C^N} \Phi''(u(z,t)) (\partial_tu(z,t))^2 \, d\nu(z) \\
            &\qquad\qquad\qquad\qquad+ \int_0^1 \Phi''(\tau) \Big( \int_{\{u(\cdot,t) > \tau\}} \partial^2_{tt} u(z,t) \, d\nu(z) \Big) \, d\tau.
        \end{align*}
        Our goal is to rewrite the first integral using the coarea formula in order to obtain the representation formula \eqref{eq:second differential at arbitrary point}. Before doing so, we notice that there exist $\varepsilon>0$ and a compact subset $K \subset \C^N \setminus \{0\}$ such that, for all $|t|<\varepsilon$ and $Y\in (T_{X_0} V)^\perp$, with $|Y|=1$, 
		\begin{equation*}
			a \leq u(z,t) \leq b \implies z \in K.
		\end{equation*}
		Indeed, since for $t=0$ we have $u(z,0) = (1+|z|^2)^{-M}$ and since $X(t)$ has unit speed, we have  $|u(z,t) - (1+|z|^2)^{-M}| \leq C\varepsilon$ for $|t|<\varepsilon$ and $Y\in (T_{X_0} V)^\perp$ with $|Y|=1$, if $\varepsilon$ is small enough, for some constant $C>0$. Hence one can take 
        \[
        K = \{z\in\C^N:\ a-C\varepsilon \leq (1+|z|^2)^{-M} \leq b+C\varepsilon\}.
        \]
        Moreover, since $\nabla u(z,0) \neq 0$ for $z \in K$, up to reducing $\varepsilon$ if necessary, we can suppose that $|\nabla u(z,t)| > c > 0$ for every $z \in K$, $|t| < \varepsilon$ and $Y\in (T_{X_0} V)^\perp$, with $|Y|=1$. This allows us to use the coarea formula (\cite[Theorem 3.13 (ii)]{evans_gariepy}) and to obtain
        \begin{equation}\label{eq derivata in 0}
            (G \circ X)''(t) = \int_0^1 \Phi''(\tau) h(\tau,t) \, d\tau,
        \end{equation}
        where (cf. \eqref{eq dnu})
        \begin{align}\label{eq def h}
        h(\tau,t)=h(\tau,t,Y)=\int_{\{u(\cdot,t) = \tau\}} &\dfrac{c_N (\partial_t u(z,t))^2}{|\nabla u(z,t)|^2 (1+|z|^2)^{N+1}} \, d\mathcal{H}^{2N-1}(z) \nonumber\\
        &\qquad\qquad\qquad +\int_{\{u(\cdot,t) > \tau\}} \partial^2_{tt} u(z,t) \, d\nu(z)\nonumber
        \end{align}
        which is the desired representation --- here $d\mathcal{H}^{2N-1}$ denotes the $(2N-1)$-dimensional Hausdorff measure. Moreover, since for  $\tau\in (a,b)$, $|t|<\varepsilon$ and $Y \in (T_{X_0}V)^{\perp}$, with $|Y|=1$, the equality $u(z,t)=\tau$ implies $z\in K$ and $|\nabla u(z,t)| > c > 0$, it is easy to see that $h$ is continuous, as claimed.

        To conclude the proof, we need to show that $h(\tau,0)$ is strictly negative for every $\tau \in (a,b)$. We observe that setting $t=0$ in \eqref{eq derivata in 0} we obtain
        \[
        (G \circ X)''(0) = \int_0^1 \Phi''(\tau)h(\tau,0)\, d\tau. 
        \]
        However, we have already computed $(G \circ X)''(0)$ in the proof of Lemma \ref{lem:second differential is negative definite}. In fact, combining equations \eqref{eq:second differential in compact form} and \eqref{eq:expression of b_alpha}, with $Y = (Y_{\alpha})_{|\alpha| \leq M} = X'(0)$  and with the change of variable $\tau = (1+s)^{-M}$ we have obtained that
        \begin{align*}
            (G \circ X)''(0) = 
            \int_0^1 \Phi''(\tau) \tilde{h}(\tau) d \tau,
        \end{align*}
        for a continuous function $\tilde{h}(\tau)=\tilde{h}(\tau,Y)$
  \footnote{Specifically, with
        \[
\tilde{h}(\tau,Y)=\sum_{\alpha \neq 0} 2|Y_{\alpha}|^2 N \tilde{c}_{\alpha}^2 \tau^{1+\frac{N}{M}} \Big[\dfrac{(1-\tau^{-1/M})^{|\alpha|+N-1}}{M}-\dfrac{1}{A_{M,N,|\alpha|}}\sum_{j=N}^{|\alpha|+N-1} \binom{M+N}{j} (1-\tau^{-1/M})^j\Big].
        \]
        }, which was found to be  strictly negative for $\tau \in (0,1)$ whenever $Y \in (T_{X_0}V)^{\perp}$. Therefore, we just need to prove that $\tilde{h}(\tau) = h(\tau,0)$ for every $\tau \in (a,b)$. This is clear, because the equality
        \begin{equation*}
            \int_0^1 \Phi''(\tau) \tilde{h}(\tau) \, d\tau = \int_0^1 \Phi''(\tau) h(\tau,0) \, d\tau
        \end{equation*}
        holds for every convex function $\Phi\in C^2([0,1])$ whose second derivative is compactly supported in $(a,b)$ and, therefore, holds for every positive Radon measure in $(a,b)$, which implies $h(\tau,0) = \tilde{h}(\tau)$ for $\tau\in (a,b)$ and $Y \in (T_{X_0}V)^{\perp}$, with $|Y|=1$.
	\end{proof}
    \section*{Acknowledgments}
    F.~N.~is a Fellow of the \textit{Accademia delle Scienze di Torino} and a member
of the \textit{Societ\`a Italiana di Scienze e Tecnologie Quantistiche (SISTEQ)}.


\end{document}